\documentclass[pra,
superscriptaddress,
amsmath,amssymb]{revtex4}

\usepackage{graphicx} 
\usepackage{subfig}
\usepackage{orcidlink}

\usepackage{graphicx}
\usepackage{amsbsy} 
\usepackage{amsfonts}
\usepackage{amsthm}
\usepackage{braket}
\usepackage{color}
\usepackage{mathrsfs}
\usepackage{appendix}
\usepackage{ragged2e}

\usepackage[normalem]{ulem}

\newtheorem*{Thm*}{Theorem}
\newtheorem{Thm}{Theorem}
\newtheorem{Lem}[Thm]{Lemma}

\newtheorem{Cor}[Thm]{Corollary}
\newtheorem{Def}[Thm]{Definition}
\newtheorem{Rem}[Thm]{Remark}
\newtheorem{algorithm}[Thm]{Algorithm}

\allowdisplaybreaks 

\begin{document}

\title{Modified Recursive QAOA for Exact Max-Cut Solutions on Bipartite Graphs: \\ Closing the Gap Beyond QAOA Limit}

\author{Eunok Bae \orcidlink{0000-0001-7531-0992}}
\email{eobae@kias.re.kr}
\affiliation{School of Computational Sciences, Korea Institute for Advanced Study (KIAS), Seoul 02455, Korea}

\author{Hyukjoon Kwon \orcidlink{0000-0001-5520-0905}}
\affiliation{School of Computational Sciences, Korea Institute for Advanced Study (KIAS), Seoul 02455, Korea}

\author{V Vijendran\orcidlink{0000-0003-3398-1821}}
\affiliation{Centre for Quantum Computation and Communication Technologies (CQC2T), Department of Quantum Science and Technology, Research School of Physics, Australian National University, Acton 2601, Australia}
\affiliation{A*STAR Quantum Innovation Center (Q.InC), Agency for Science, Technology and Research (A*STAR),
2 Fusionopolis Way, Innovis \#08-03, Singapore 138634, Republic of Singapore}

\author{Soojoon Lee\orcidlink{0000-0003-2925-1017}}
\email{level@khu.ac.kr}
\affiliation{Department of Mathematics and Research Institute for Basic Sciences, Kyung Hee University, Seoul 02447, Korea}
\affiliation{School of Computational Sciences, Korea Institute for Advanced Study (KIAS), Seoul 02455, Korea}

\date{\today}

\begin{abstract}
 Quantum Approximate Optimization Algorithm (QAOA) is a quantum-classical hybrid algorithm proposed with the goal of approximately solving combinatorial optimization problems such as the MAX-CUT problem. It has been considered a potential candidate for achieving quantum advantage in the Noisy Intermediate-Scale Quantum era and has been extensively studied. However, the performance limitations of low-level QAOA have also been demonstrated across various instances. In this work, we first analytically prove the performance limitations of level-1 QAOA in solving the MAX-CUT problem on bipartite graphs. To this end, we derive an upper bound for the approximation ratio based on the average degree of bipartite graphs. 
 Second, we demonstrate that Recursive QAOA (RQAOA), which recursively reduces graph size using QAOA as a subroutine, outperforms the level-1 QAOA. However, the performance of RQAOA exhibits limitations as the graph size increases.
 Finally, we show that RQAOA with a restricted parameter regime can fully address these limitations. Surprisingly, this modified RQAOA always finds the exact maximum cut for any bipartite graphs and even for a more general graph with parity-signed weights.
\end{abstract}

\maketitle

\section{Introduction}
  Variational Quantum Algorithms (VQAs)~\cite{MAR+21} have emerged as a practical solution for harnessing quantum computation in the Noisy Intermediate-Scale Quantum (NISQ) era~\cite{Pre18, MBB+18}. NISQ devices are characterized by their relatively small number of qubits, their susceptibility to noise, and a lack of error correction, which makes them impractical for running traditional deep quantum algorithms like Shor's algorithm~\cite{Shor} designed for fault-tolerant quantum systems. VQAs were introduced to address these limitations, employing a hybrid quantum-classical framework~\cite{PMS+14,FGG14,MBB+18,MAR+21,TCC+22}. In these algorithms, the quantum part prepares a parameterized quantum state and measures observables, while the classical part optimizes these parameters to either minimize or maximize an objective function, such as an energy eigenvalue or a cost function. VQAs provide a powerful and practical tool for solving various problems in the NISQ era, from quantum chemistry to optimization and quantum machine learning~\cite{PMS+14,FGG14,MBB+18,SK19,MAR+21}.
 
 One of the most prominent examples of VQAs is Quantum Approximate Optimization Algorithm (QAOA)~\cite{FGG14}. QAOA is designed to tackle combinatorial optimization problems, where the goal is to find an approximate solution that is close to the optimal one. The algorithm works by alternating between applying a problem-specific operator and mixing operators in a parameterized quantum circuit. The parameters of the circuit are optimized using a classical optimizer to improve the solution iteratively. QAOA has gained significant attention for its potential to outperform classical algorithms in solving combinatorial optimization problems like the MAX-CUT problem, where it aims to find a cut that maximizes the number of edges between two partitions of vertices of a graph~\cite{FGG14,WHJR18,MAR+21,ZWC+20}. 
 Despite its promise for near-term quantum applications, it has been established that QAOA, when applied at any constant level, exhibits limited performance in solving the MAX-CUT problem for several instances~\cite{Has19,BKKT19,FGG20,Mar21,BM22}.
 Several variants and extensions of QAOA have been proposed to address its limitations and enhance its performance across a broader range of problems~\cite{BKKT19,HWO+19,EMW21,HLO+22,ZTB+22,VDKAL24}.

 Recursive QAOA (RQAOA) is one of the variants of QAOA proposed to enhance the performance on complex instances and scales to larger problems by reducing the problem size~\cite{BKKT19}. 
 The performance of QAOA can be constrained by the $Z_2$ symmetry of its quantum states and the geometric locality of the ansatz, which means that the cost operators only involve interactions between nearest neighbor qubits in the underlying graph~\cite{BKKT19}. To address these limitations, RQAOA reduces the problem size iteratively by fixing certain variables based on edge correlation analysis. This iterative reduction can lead to more precise solutions as the recursion progresses through the new connections between previously unlinked vertices, effectively counteracting the locality constraints of QAOA. 
 Although RQAOA has been less extensively investigated compared to other QAOA variants, it is gaining increasing attention as a promising approach for NISQ devices~\cite{BKKT19,BGGS21,BKKT22,PJBD24,BL24,FPW+24}. 
 
 It has been shown that RQAOA outperforms the original QAOA for many problem instances. Numerical evidence on random graphs shows that the level-1 RQAOA significantly outperforms QAOA at the same level for solving the MAX-CUT problem~\cite{BKKT19}. Another experimental evidence found by the same authors demonstrates that the level-1 RQAOA is even competitive with Newman's classical algorithm for general $k$-colorable graphs for solving the MAX-$k$-CUT problem (especially, when $k=3$) which is a generalization of MAX-CUT, where the goal is to partition the vertices of a graph into $k$ disjoint subsets such that the number of edges between different subsets is maximized~\cite{BKKT22}. Furthermore, there have been rigorous proofs supporting this argument that RQAOA surpasses the original QAOA in certain special cases. It has been proved that the level-1 RQAOA achieves the optimal approximation ratio while that of any constant level-$q$ QAOA with $q<\frac{n}{4}$ is less than $1-\frac{1}{4q+2}$ in solving the MAX-CUT problem on cycle graphs, where $n$ is the number of vertices of a cycle graph~\cite{BKKT19}. Recently, a separation of the performance between RQAOA and QAOA for solving the MAX-CUT problem on complete graphs with $2n$ vertices has been proved~\cite{BL24}. It was analytically shown that in this specific case, the level-1 RQAOA achieves an approximation ratio of 1, whereas the approximation ratio of the original level-1 QAOA is strictly upper-bounded as $1-\frac{1}{8n^2}$. 
 
 Our first significant result in this work supports the claim that RQAOA performs better than QAOA. We prove the performance limitations of the level-1 QAOA for solving the MAX-CUT problem on bipartite graphs in terms of the approximation ratio. Among the two bounds we derive in Theorem~\ref{thm:bipar}, the first bound is shown to be tight for complete bipartite graphs, while the second bound, though slightly looser, is obtained using the average degree which is a key characteristic representing the graph's density. Our results indicate that, intuitively, the performance of the level-1 QAOA for solving the MAX-CUT problem tends to degrade as a bipartite graph becomes denser and as a complete bipartite graph becomes larger, respectively. In particular, as the size of the complete bipartite graph increases, the performance of the level-1 QAOA becomes worse approaching that of random guessing, in contrast to the case of complete graphs, which converge to one~\cite{BL24}.

 Unlike other provable cases such as cycle graphs~\cite{BKKT19} and complete graphs~\cite{BL24}, in bipartite graphs, as the reduction process of RQAOA progresses, the structure of the graph deteriorates, causing it to lose its bipartite nature. This issue becomes more pronounced in graphs with a large number of vertices, where the graph transforms into a weighted graph. Consequently, analytically proving the performance of RQAOA in such scenarios remains a challenging problem.
 Although it has not yet been proved whether RQAOA can achieve the optimal performance in solving the MAX-CUT problem on bipartite graphs, we numerically show that the level-1 RQAOA has a much better performance than the level-1 QAOA for this problem instance in Sec~\ref{sec:rqaoa}. 

 While RQAOA offers crucial improvements over the original QAOA, it also faces several limitations as identified in recent research~\cite{PJBD24} and a few works have generalized and extended RQAOA to improve its performance~\cite{FKS+24,BH23,WNL23,PJBD24}. 
 From the numerical results we obtain for random weighted bipartite graphs with 128 and 256 vertices in Sec~\ref{sec:rqaoa}, RQAOA appears to fall short of achieving optimal performance. Although employing a better optimization method might lead to improved outcomes, it remains unclear whether this issue stems from the limitations of RQAOA itself in determining edge correlations or from challenges in the optimization process during QAOA iterations. 

 In this paper, we propose a parameter setting method to enhance the performance of RQAOA and prove that our modified RQAOA can always find the exact MAX-CUT solution for positive weighted bipartite graphs. 
 Although similar parameter setting schemes applied to QAOA have been proposed~\cite{SHS+24,SLL+23,LG24} and they have focused on finding optimal parameters, 
 our contribution lies in rigorously proving that modifying the parameter search space to preserve the graph structure significantly improves RQAOA performance for certain cases, regardless of whether optimal parameters are found during QAOA iterations.

This paper is organized as follows. In Sec.~\ref{sec:pre}, we introduce the MAX-CUT problem, which we aim to solve, and provide a brief overview of the quantum algorithms used to tackle this problem: QAOA and RQAOA. In Sec.~\ref{sec:lim_qaoa}, we analytically prove that the level-1 QAOA has limited performance in solving the MAX-CUT problem on bipartite graphs, deriving an upper bound for the approximation ratio based on the average degree of the graph. In Sec.~\ref{sec:rqaoa}, we numerically show that 
while the level-1 RQAOA outperforms the level-1 QAOA in solving the MAX-CUT problem on bipartite graphs, it also fails to find the exact solution when the graph size increases. We propose a modified RQAOA to improve its performance and prove that our modified RQAOA can exactly solve the MAX-CUT problem on not only positive weighted bipartite graphs but also more general weighted graphs called a graph with parity-signed weights ${G=(V,E,w,\sigma)}$ in Sec.~\ref{sec:new_rqaoa}. Finally, we summarize our results and discuss the potential applicability of our modified RQAOA to other instances in Sec.~\ref{sec:conclusion}.

\section{Main Results}
\label{sec:main}
In this section, we preview our main results. Our first main result is to prove the performance limitations of the level-1 QAOA in solving the MAX-CUT problem on bipartite graphs. The approximation ratio is upper bounded in terms of the average degree which is a key property representing the graph's density. More details can be found in Sec~\ref{sec:lim_qaoa}.

\vspace{0.3cm}
\textbf{Theorem~\ref{thm:bipar}.}
\label{thm:bipar}
\textit{(Limitation of QAOA$_1$ for solving MAX-CUT on bipartite graphs) Let $G=(V,E)$ be a bipartite graph and $\alpha_1$ be the approximation ratio of the level-1 QAOA for solving the MAX-CUT problem on $G$. Then
$$
\alpha_1 \le \frac{1}{2} + \sum_{d=1}^{d_{\max}} \frac{d |D_d|}{2|E|} \left[ \frac{1}{2\sqrt{d}}\left(\frac{\sqrt{d-1}}{\sqrt{d}}\right)^{d-1} \right] \le \frac{1}{2} + \frac{1}{2\sqrt{e}}\left( \frac{1}{\sqrt{d_{ave}}} + \frac{\sqrt{e}-1}{d_{ave}} \right),
$$
where $D_d$ is the set of vertices with degree $d$ and $d_{ave} = \frac{2|E|}{|V|}$ denotes the average degree of $G$.}

\vspace{0.3cm}
In Sec.~\ref{sec:rqaoa}, we numerically show that the level-1 RQAOA outperforms the level-1 QAOA for solving the MAX-CUT problem on weighted bipartite graphs, and we suspect that RQAOA might not be able to achieve the optimal approximation ratio in this problem instance. Furthermore, it appears that there is an optimization issue in the QAOA subroutine that needs to be addressed for RQAOA to solve the problem with a better approximation ratio. In this vein, we propose a modified RQAOA with a parameter setting strategy to enhance the performance for this problem instance as follows.

\vspace{0.3cm}
\textbf{Algorithm~\ref{alg:modified_rqaoa}.} 
(Modified RQAOA$_1$ equipped with a parameter setting method)
We only modify the first step of RQAOA$_1$ and the rest part of our algorithm is the same as the original RQAOA$_1$.
\begin{enumerate}
    \item Apply the level-$1$ QAOA to find the optimal parameters $({\beta}^*,{\gamma}^*)$ to maximize $F_1(\mathbf{\beta},\mathbf{\gamma})$ in the restricted domain where $0<|\gamma| \le \frac{\pi}{2 w^*}$ and $0<\beta<\frac{\pi}{4}$ with $w^* = \max_{(i,j) \in E} |w(ij)|$.
\end{enumerate}

\vspace{0.3cm}
Our second main result is to rigorously prove that this modified RQAOA can achieve the approximation ratio of 1 in solving the MAX-CUT problem on a graph with parity-signed weights which is a generalization of bipartite graphs as shown in Sec.~\ref{sec:new_rqaoa}.

\vspace{0.3cm}
\textbf{Theorem~\ref{thm:main}.}
\textit{Our Algorithm~\ref{alg:modified_rqaoa} can exactly solve the MAX-CUT problem on a graph with parity-signed weights.}


\section{Premilinaries}
\label{sec:pre}
\subsection{MAX-CUT problem}
Let $G=(V,E)$ be a (undirected) graph with the set of vertices $V$ and the set of edges ${E=\{(i,j):i,j \in V\}}$. The MAX-CUT problem is a well-known combinatorial optimization problem that aims to split $V$ into two disjoint subsets such that the number of edges spanning the two subsets is maximized. The MAX-CUT problem can be formulated by maximizing the cost function 
 \[
 C(\mathbf{x})=\frac{1}{2} \sum_{(i,j)\in E} \left(1-x_i x_j \right)
 \]
 for $\mathbf{x}=(x_1,x_2,\dots,x_n) \in \{-1,1\}^n$.

\subsection{QAOA}

QAOA can be viewed as a discrete version of adiabatic quantum computing~\cite{FGG14}. The design of QAOA involves constructing a parameterized quantum circuit that alternates between applying the problem Hamiltonian $H_C$ (which encodes the optimization problem) and the driving Hamiltonian $H_B$ (which ensures broad exploration of the solution space). Here, we only focus on the MAX-CUT problem which can be converted to the following problem Hamiltonian
 \[
 H_C=\frac{1}{2} \sum_{(i,j)\in E} \left( I-Z_i Z_j \right),
 \]
 where $Z_i$ is the Pauli operator $Z$ acting on the $i$-th qubit.
The level-$q$ QAOA, denoted by QAOA$_q$, can be described as the following algorithm.
 \begin{algorithm}[QAOA$_{q}$~\cite{FGG14}]
The QAOA$_{q}$ is as follows.
 \begin{enumerate}
     \item Prepare the initial state $\ket{+}^{\otimes n}$.
     \item Generate a variational wave function
     \[\ket{\psi_q(\mathbf{\beta},\mathbf{\gamma})}
     =\prod_{j=1}^q e^{-i\beta_jH_B}e^{-i\gamma_jH_C}
     \ket{+}^{\otimes n},
     \]
     where $\mathbf{\beta}=(\beta_1,\ldots,\beta_q)$,
     $\mathbf{\gamma}=(\gamma_1,\ldots,\gamma_q)$, 
     $H_C$ is a problem Hamiltonian, 
     $H_B=\sum_{i=1}^n X_i$ is a driving Hamiltonian, 
     and $X_i$ is the Pauli operator $X$ acting on the $i$-th qubit.
     \item Compute the expectation value 
     \[F_q(\mathbf{\beta},\mathbf{\gamma})=\bra{\psi_q(\mathbf{\beta},\mathbf{\gamma})}H_C\ket{\psi_q(\mathbf{\beta},\mathbf{\gamma}})=\left< H_C\right>\] 
     by performing the measurement in the computational basis.
     \item Apply a classical optimization algorithm to find the optimal parameters 
     \[(\mathbf{\beta}^*,\mathbf{\gamma}^*)= \mathrm{argmax}_{\mathbf{\beta},\mathbf{\gamma}} F_q(\mathbf{\beta},\mathbf{\gamma}).\]
 \end{enumerate}
 \end{algorithm}
  
The approximation ratio $\alpha$ of QAOA$_{q}$ is defined as \[\alpha_q=\frac{F_q(\mathbf{\beta}^*,\mathbf{\gamma}^*)}{C_{\max}},\] 
where $C_{\max}=\max_{\mathbf{x}\in \{-1,1\}^n}C(\mathbf{x})$.

\subsection{Recursive QAOA}

In this section, we provide a brief overview of RQAOA which was introduced to address the limitations of the original QAOA by incorporating a recursive problem reduction strategy~\cite{BKKT19}. RQAOA iteratively reduces the problem size by fixing specific variables based on edge correlation analysis. As the recursion proceeds, this iterative reduction can yield more accurate solutions by establishing new connections between previously unconnected vertices, thereby counteracting the locality constraints of QAOA. For a level-$q$ RQAOA, denoted by RQAOA$_q$, we consider an Ising-like Hamiltonian, which is closely related the weighted MAX-CUT Hamiltonian, given by
$$
H_n=\sum_{(i,j) \in E}J_{i,j}Z_iZ_j,
$$
defined on a graph $G_n=(V,E)$ with $|V|=n$, where $J_{i,j}\in \mathbb{R}$ are arbitrary.
The goal of RQAOA is to approximate 
$$
\max_{\mathbf{x} \in \{-1,1\}^n} \bra{\mathbf{x}}H_n\ket{\mathbf{x}}.
$$
Here, $\ket{\mathbf{x}}=\ket{x_1,\dots,x_n}$ and  
$Z\ket{x_i}=x_i\ket{x_i}$ for each $i=1,\dots,n$.
RQAOA$_q$ can be described as the following algorithm.

\begin{algorithm}[RQAOA$_{q}$~\cite{BKKT19}]
\label{alg:rqaoa}
The level-$q$ RQAOA is as follows.
\begin{enumerate}
    \item Perform the original QAOA$_{q}$ to find the optimal parameters $({\beta}^*,{\gamma}^*)$ to maximize $F_q(\mathbf{\beta},\mathbf{\gamma})$.
    \item Compute the edge expectation values $$M_{ij}=\bra{\psi_q(\mathbf{\beta}^*,\mathbf{\gamma}^*)}Z_iZ_j\ket{\psi_q(\mathbf{\beta}^*,\mathbf{\gamma}^*)}$$ for every edges $(i,j)\in E$.
    \item Pick the edge $(k,l)  =  \mathrm{argmax}_{(i,j) \in E} M_{ij}$
    \item By imposing the constraint $Z_k=\textrm{sgn}(M_{kl})Z_l$, replace it with $H_n$ to obtain 
    \begin{eqnarray*}
    H_n'    =
    \textrm{sgn}(M_{kl})\left[\sum_{(i,k) \in E}J_{ik}Z_iZ_l 
    \right] +    \sum_{i,j\neq k}J_{ij}Z_iZ_j
    \end{eqnarray*}
    \item Call the QAOA recursively to maximize the expected value of a new Ising Hamiltonian $H_{n-1}$ depending on $n-1$ variables:
    \[
    H_{n-1}=\sum_{(i,l) \in E'_0}J'_{ij}Z_iZ_l +\sum_{(i,j) \in E'_1}J'_{ij}Z_iZ_j, 
    \]
    where  
    $$
    E'_0=\{(i,l) : (i,k) \in E\}, 
    E'_1=\{(i,j) : i,j \neq k\}, 
    $$
    and 
    \[
    J'_{ij}=
    \begin{cases}
    \mathrm{sgn}(M_{kl})J_{ik}
     & \mathrm{if}~(i,l) \in E'_0, \\
    J_{ij} & \mathrm{if}~(i,j) \in E'_1.
    \end{cases}
    \]
    \item The recursion stops when the number of variables reaches some threshold value $n_c \ll n$, and find $\mathbf{x}^*=\mathrm{argmax}_{\mathbf{x}\in \{-1,1\}^{n_c}}\bra{\mathbf{x}}H_{n_c}\ket{\mathbf{x}}$ by a classical algorithm.
    \item Reconstruct the original (approximate) solution $\tilde{\mathbf{x}}\in\{-1,1\}^n$ from $\mathbf{x}^*$ using the constraints.
\end{enumerate}
 \end{algorithm}

\section{Limitation of QAOA$_{1}$}
\label{sec:lim_qaoa}

In this section, we analytically prove that the level-1 QAOA has limited performance in solving the MAX-CUT problem on bipartite graphs. We derive the upper bound of the approximation ratio in terms of the average degree of a given graph.

The analytical form for the expectation value of the level-$1$ QAOA has been known in Ref.~\cite{WHJR18}. For all $(i,j) \in E$, 
\begin{eqnarray}
\label{eq:Cij}
\left< C_{ij} \right> &=& \left<\psi_1(\beta, \gamma) \right|\frac{1}{2}\left( I - Z_iZ_j\right) \left| \psi_1(\beta, \gamma) \right> \\ \nonumber
&=& \frac{1}{2} - \frac{1}{4}\sin^2(2\beta)\cos^{d_i+d_j-2(n_{ij}+1)}(1-\cos^{n_{ij}}(2\gamma)) + \frac{1}{4}\sin(4\beta) \sin\gamma \left(\cos^{d_i-1}\gamma + \cos^{d_j-1}\gamma \right),
\end{eqnarray}
where $d_i$ is the degree of the vertex $i$ and $n_{ij}$ is the number of triangles with the edge $(i,j)$. 
Since bipartite graphs are triangle-free, Eq.~(\ref{eq:Cij}) becomes simpler. In this case, we can rewrite the expectation value of the MAX-CUT Hamiltonian of the level-1 QAOA in terms of the vertex degrees instead of the edges as follows.
    \begin{eqnarray*}
        \left< H_C \right> 
        &=& \sum_{(i,j)\in E} \left[\frac{1}{2} + \frac{1}{4}\sin(4\beta) \sin\gamma \left(\cos^{d_i-1}\gamma + \cos^{d_j-1}\gamma \right) \right] \\
        &=& \frac{1}{2}|E|+ \frac{1}{4} \sin(4\beta)\sin\gamma \sum_{d=1}^{d_{\max}} d |D_d| \cos^{d-1}\gamma,
    \end{eqnarray*}
    where $D_d$ is the set of vertices with degree $d$. For bipartite graphs, we know the optimal cut is the number of all edges. Thus, the approximation ratio is
    \begin{eqnarray*}
        \alpha_1 &=& \max_{\beta,\gamma}\left< H_C \right> / C_{\max} \\
        &=& \max_{\gamma} \left[\frac{1}{2} + \frac{1}{2} \sum_{d=1}^{d_{\max}} \frac{d |D_d|}{2|E|} \sin\gamma \cos^{d-1}\gamma \right] \\ 
        &=& \max_{\gamma} \left[\frac{1}{2} + \frac{1}{2} \sum_{d=1}^{d_{\max}} \frac{d |D_d|}{2|E|}  f_d(\gamma) \right],
    \end{eqnarray*}
    where $f_d(\gamma)=\sin\gamma\cos^{d-1}\gamma$.
    Then we can get a bound of $\alpha_1$ by obtaining the upper bound of $f_d(\gamma)$ for each~$d$. Although this bound is obtained quite naturally, but it falls short in capturing the properties of a given bipartite graph. Therefore, by proposing a new bound that utilizes the average degree $d_{ave}$, which reflects one of the important properties of a graph—its density—we can intuitively observe how the density of bipartite graphs impacts the performance of the level-1 QAOA for solving the MAX-CUT problem.
    
\begin{Thm}[Bipartite graph]
\label{thm:bipar}
Let $G=(V,E)$ be a bipartite graph and let $\alpha_1$ be the approximation ratio of the level-1 QAOA for solving the MAX-CUT problem on $G$. Then
$$
\alpha_1 \le \frac{1}{2} + \sum_{d=1}^{d_{\max}} \frac{d |D_d|}{2|E|} \left[ \frac{1}{2\sqrt{d}}\left(\frac{\sqrt{d-1}}{\sqrt{d}}\right)^{d-1} \right] \le \frac{1}{2} + \frac{1}{2\sqrt{e}}\left( \frac{1}{\sqrt{d_{ave}}} + \frac{\sqrt{e}-1}{d_{ave}} \right),
$$
where $d_{ave}$ denotes the average vertex degree of $G$ which can be defined as $\frac{\sum_{v \in V}d_v}{|V|}$, or equivalently, $\frac{2|E|}{|V|}$.
\end{Thm}

\begin{proof}
In order to get the upper bound of $\alpha_1$, let us first find the optimal value of $f_d(\gamma)$ as follows. Since $f_1(\gamma)=\sin\gamma$, $f_1(\gamma) \le 1$. Observe that $f_d(\gamma)=\sin\gamma\cos^{d-1}\gamma$ and $f'_d(\gamma)=\cos^{d-2}\gamma(d\cos^2\gamma-(d-1))$ for each $d \ge 2$. Then we can find the optimal $\gamma^*$ such that $\cos\gamma^*=(1-\frac{1}{d})^{1/2}$. 
Thus, we have
$$f_d(\gamma) \le f_d(\gamma^*) = \frac{1}{\sqrt{d}}\left(\frac{\sqrt{d-1}}{\sqrt{d}}\right)^{d-1},$$ 
and naturally get the upper bound of the approximation ratio as follows. 
\begin{eqnarray*}
    \alpha_1 
    &\le& \frac{1}{2} + \sum_{d=1}^{d_{\max}} \frac{d |D_d|}{2|E|} \left[ \frac{1}{2\sqrt{d}}\left(\frac{\sqrt{d-1}}{\sqrt{d}}\right)^{d-1} \right] \\
    &\le& \frac{1}{2} + \frac{1}{2}\sum_{d=1}^{d_{\max}} \frac{d |D_d|}{2|E|} \left[ \sqrt{\frac{1}{ed}}+\frac{1}{d} \left( 1-\frac{1}{\sqrt{e}}\right)\right] \\
    &\le& \frac{1}{2} + \frac{1}{2} \sqrt{\sum_{d=1}^{d_{\max}} \frac{|D_d|}{2e|E|}} + \sum_{d=1}^{d_{\max}} \frac{|D_d|}{2|E|}\left( 1-\frac{1}{\sqrt{e}}\right)  \\
    &\le& \frac{1}{2} + \frac{1}{2\sqrt{e}}\left( \frac{1}{\sqrt{d_{ave}}} + \frac{\sqrt{e}-1}{d_{ave}} \right),
\end{eqnarray*}
where $d_{ave}$ denotes the average vertex degree of $G$ which can be defined as $\frac{\sum_{v \in V}d_v}{|V|}$, or equivalently, $\frac{2|E|}{|V|}$, or equivalently, $\frac{2|E|}{\sum_d |D_d|}$. 
To show that the second inequality holds, we can prove that the following inequalities hold separately
$$
 \frac{1}{\sqrt{d}}\left(\frac{\sqrt{d-1}}{\sqrt{d}}\right)^{d-1} \le \frac{1}{\sqrt{ed+1-e}} 
 ~~\text{and}~~ \frac{1}{\sqrt{ed+1-e}} \le \frac{1}{\sqrt{e}d}\left(\sqrt{d} + \sqrt{e} -1 \right)
$$
by showing that
$$
\ln \left[\frac{1}{\sqrt{d}}\left(\frac{\sqrt{d-1}}{\sqrt{d}}\right)^{d-1}\right] \le \ln \left(\frac{1}{\sqrt{ed+1-e}} \right)
 ~~\text{and}~~ ed^2 \le (ed+1-e)\left(\sqrt{d} + \sqrt{e} -1\right)^2.
$$

It follows from the concavity of the function
$$
g(x) := \frac{1}{\sqrt{e}}\sqrt{x} + \left( 1-\frac{1}{\sqrt{e}}\right) x 
$$
that the third inequality holds.  
\end{proof}

\begin{figure}[!ht]
  \centering
  \subfloat[][Bipartite graphs with $p=0.6$]{\includegraphics[width=.4\textwidth]{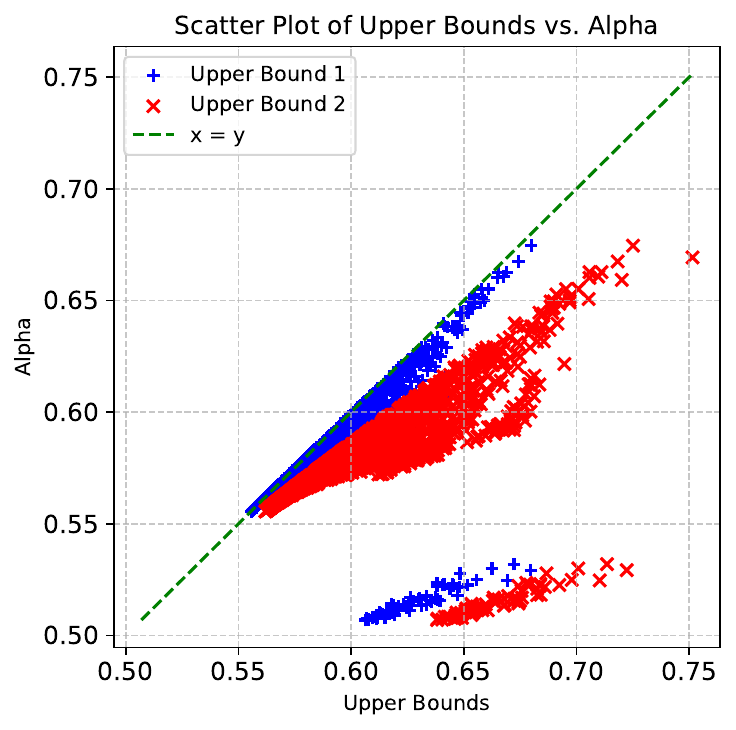}}\quad
  \subfloat[][Bipartite graphs with $p=0.9$]{\includegraphics[width=.4\textwidth]{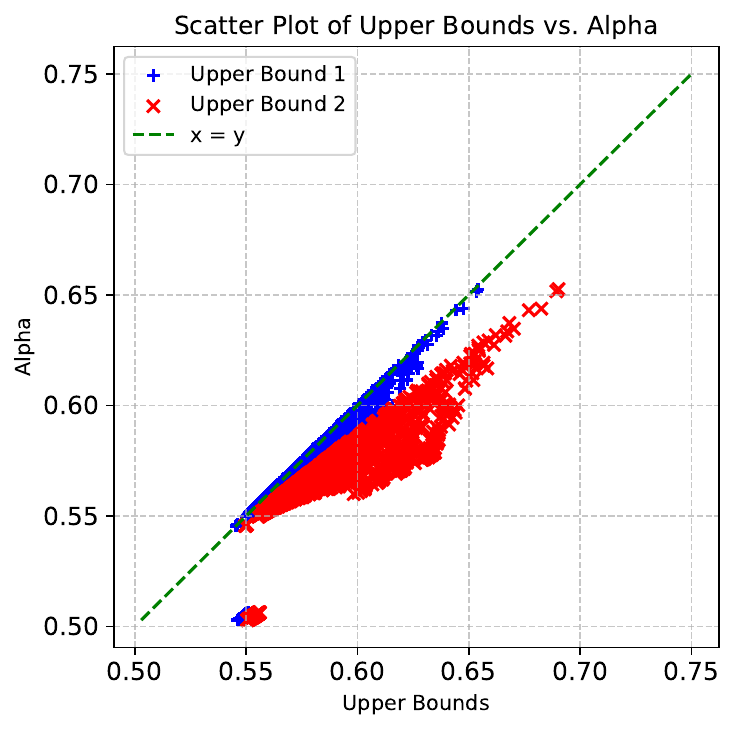}}\\
  \subfloat[][Complete bipartite graphs $K_{n,n}$]{\includegraphics[width=.4\textwidth]{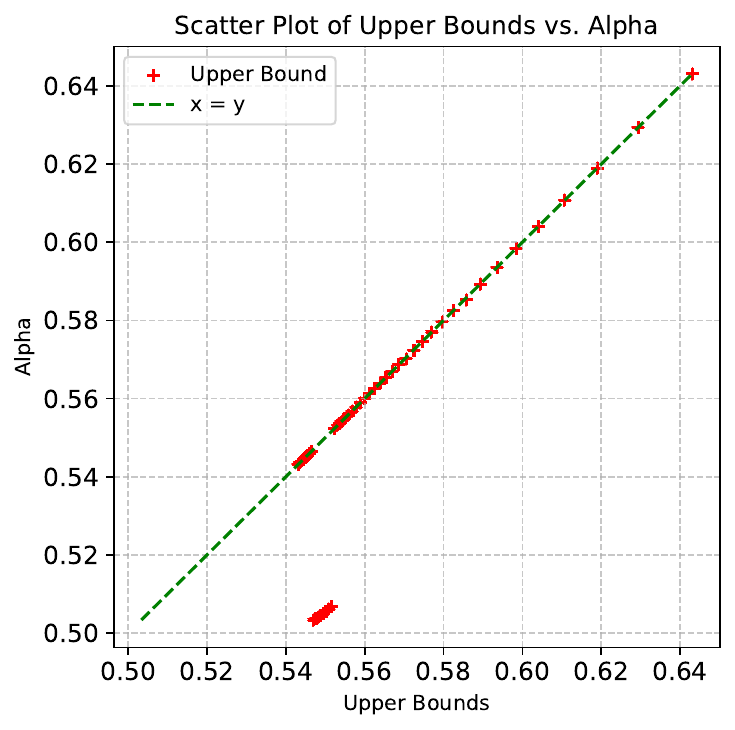}}\quad
  \subfloat[][Complete bipartite graphs $K_{n,m}$]{\includegraphics[width=.4\textwidth]{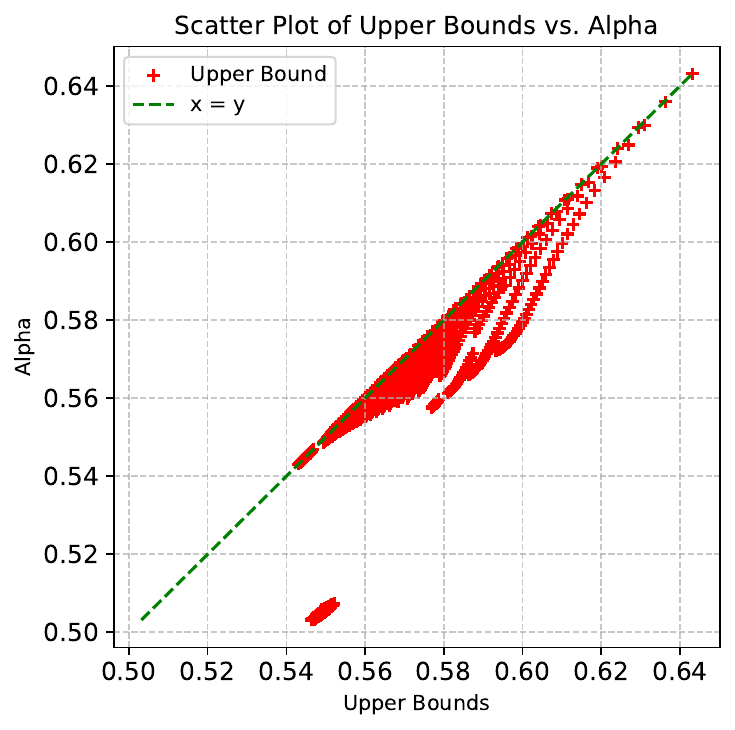}}
\caption{\justifying Plots illustrating the tightness of the bounds in Theorem~\ref{thm:bipar} and Corollary~\ref{cor}. In (a) and (b), bipartite graphs with two disjoint sets of $n,m$ vertices. Here, $n,m$ were randomly selected from $5$ to $51$ for $p=0.6$ and $p=0.9$, respectively, where $p$ represents the edge probability. The level-$1$ QAOA approximation ratio $\alpha_1$ was compared with two upper bounds provided in Theorem~\ref{thm:bipar}. Compared to (a), the graph instances in (b) are denser, and the bounds are tighter. In (c) and (d), the level-$1$ QAOA approximation ratio $\alpha_1$ is compared with the bound provided in Corollary~\ref{cor} for complete bipartite graphs 
$K_{n,n}$ and $K_{n,m}$, respectively. In this case, it is confirmed that the bounds are tight.
The points plotted below the boundary appear to have achieved relatively low approximation ratios compared to the upper bound due to numerical optimization issues.
}
  \label{Fig0}
\end{figure}

In Fig.~\ref{Fig0}, $p$ indicates the edge probability so that the number of edges is $|E|=nmp$. Thus, the graph instances in (b) have more edges compared to (a), that is, they have a larger $d_{ave}$ value, representing density, and the upper bounds become tighter compared to (a) as shown in Fig.~\ref{Fig0}. 
Although the upper bounds we have obtained in Theorem~\ref{thm:bipar} are not tight, they become tighter as a bipartite graph becomes denser which means that $d_{ave}$ increases. In particular, we can get a tight bound for the complete bipartite graphs. Moreover, the intuitive conclusion that can be drawn from the following inequality is that as the size of the complete bipartite graph increases, the performance of the level-1 QAOA deteriorates, approaching that of random guessing.
\begin{Cor}[Complete bipartite graph]
\label{cor}
    Let $K_{n,m}$ be a complete bipartite graph that has two partitioned subsets $V_1$ and $V_2$ of vertices with $|V_1|=n$ and $|V_2|=m$.
    \begin{itemize}
        \item For $n=m \ge 2$,
    $\alpha_1 \le \frac{1}{2}+\frac{1}{2\sqrt{n-1}}\left(1-\frac{1}{n} \right)^{n/2}$.
        \item For $n \neq m \ge 2$, 
        $\alpha_1 \le \frac{1}{2}+\frac{1}{4\sqrt{n-1}}\left(1-\frac{1}{n} \right)^{n/2} + \frac{1}{4\sqrt{m-1}}\left(1-\frac{1}{m} \right)^{m/2} 
        .$
    \end{itemize}
\end{Cor}

\begin{proof}
    For the case of $K_{n,n}$ with $n \ge 2$, all vertices has degree $n$ and so, $|D_n|=2n$. Then $\alpha_1 = \max_{\gamma} \left(\frac{1}{2} + \frac{1}{2}f_n(\gamma)\right)$. Observe that  
    $f_n(\gamma)=2\sin\gamma\cos^{n-1}\gamma$ and $f'_n(\gamma)=2\cos^{n-2}\gamma(n\cos^2\gamma-(n-1))$. Then we can find the optimal $\gamma^*$ such that $\cos\gamma^*=(1-\frac{1}{n})^{1/2}$ and $f_n(\gamma^*) = \frac{2}{\sqrt{n}}\left(\frac{\sqrt{n-1}}{\sqrt{n}}\right)^{n-1}$. Thus, 
    $$\alpha_1 \le \frac{1}{2} + \frac{1}{2}f_n(\gamma^*) = \frac{1}{2} + \frac{1}{2\sqrt{n-1}}\left( 1-\frac{1}{n}\right)^{n/2}.$$

    Similarly, for the case of $K_{n,m}$ with $n \neq m \ge 2$, $n$ vertices have degree $m$ and $m$ vertices have degree $n$, that is, $|D_n|=m$ and $|D_m|=n$. So, $\alpha_1 = \max_{\gamma} \left[\frac{1}{2} + \frac{1}{4}\left( f_n(\gamma) + f_m(\gamma)\right) \right]$ and thus, we get 
    \begin{eqnarray*}
        \alpha_1 &\le& \frac{1}{2} + \frac{1}{4}\left( f_n(\gamma^*) + f_m(\gamma^*)\right) \\
        &=& \frac{1}{2}+\frac{1}{4\sqrt{n-1}}\left(1-\frac{1}{n} \right)^{n/2} + \frac{1}{4\sqrt{m-1}}\left(1-\frac{1}{m} \right)^{m/2}.
    \end{eqnarray*}
\end{proof}

Remark that the upper bounds derived above are based on optimized parameters. This result implies that, even with optimal parameters, the performance of the level-1 QAOA remains fundamentally limited. 
However, the performance of the level-1 RQAOA can be influenced by parameter optimization. In the following section, we numerically demonstrate that while RQAOA outperforms QAOA, it may still fail to find the exact MAX-CUT solution on weighted bipartite graphs. This suggests that the performance of RQAOA could be affected by the parameter optimization inherent in the QAOA subroutine.
To address these challenges, we introduce a modified RQAOA with a parameter-setting strategy and prove that it can exactly solve the problem on a graph with parity-signed weights, even if global optimality is not guaranteed at each QAOA iteration.

\section{Improving performance using RQAOA}

\subsection{RQAOA Performance: Superior to QAOA}
\label{sec:rqaoa}

In this section, we numerically show a separation of the performance between the level-1 RQAOA and the level-1 QAOA in solving the MAX-CUT problem on weighted bipartite graphs with two disjoint sets of vertices, each containing $n$ vertices, denoted by $G_{n,n}^w$, as shown in Fig.~\ref{fig:rqaoa}. 

\begin{figure}[!ht]
  \centering
  \subfloat[][Benchmark Results for $G_{64, 64}^w$]{\includegraphics[width=.48\textwidth]{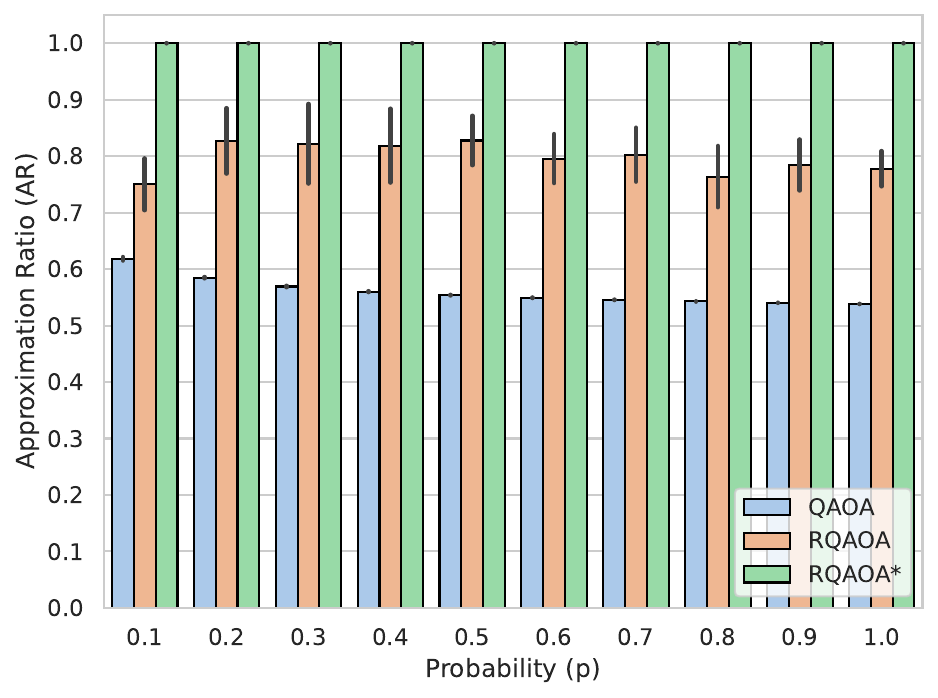}}
  \subfloat[][Benchmark Results for $G_{128, 128}^w$]{\includegraphics[width=.48\textwidth]{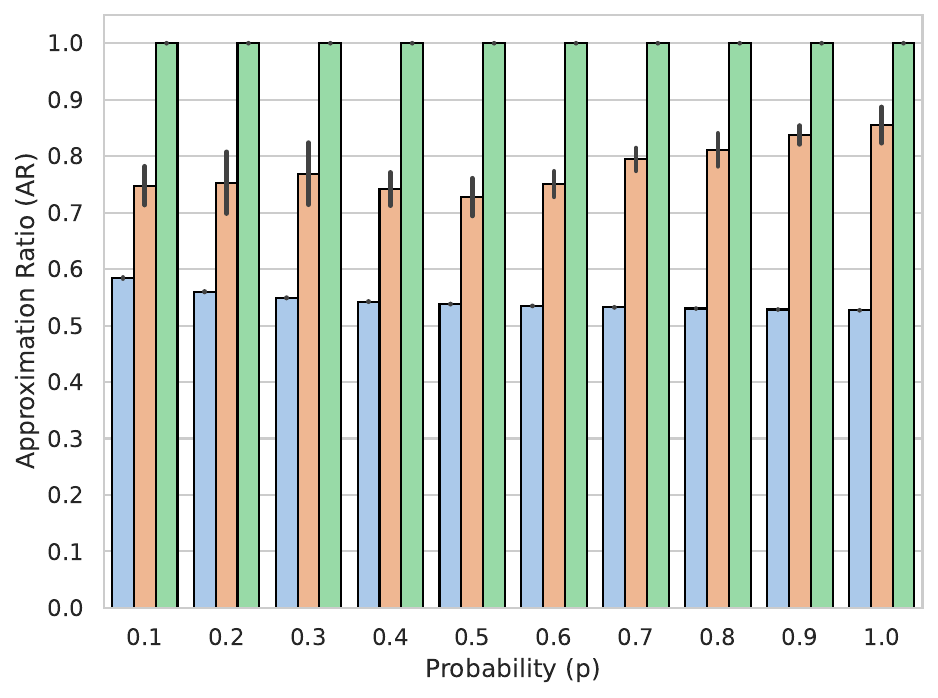}}
\caption{\justifying 
Benchmark results comparing the level-1 QAOA, RQAOA, and RQAOA$^*$ (our modified RQAOA introduced in Algorithm~\ref{alg:modified_rqaoa}) on weighted bipartite graphs $G_{64,64}^w$ and $G_{128,128}^w$ for various probabilities $p$ ranging from 0.1 to 1.0 in increments of 0.1. Here, $p$ represents the edge probability.
Edge weights were sampled from a Gaussian distribution with a mean of 50 and a variance of 25, subsequently converted to integers. Each benchmark was performed on 10 random instances per $p$.
The QAOA algorithm was optimized through a line search on $\gamma \in [0,\pi]$ with 20 samples, followed by gradient descent, with $\beta^*=\pi/8$. 
The RQAOA was optimized using a coarse grid search with $20 \times 20$ points, followed by gradient descent on the best-found point.
The results indicate that RQAOA$^*$ consistently achieves an approximation ratio of 1 across all 200 instances, while RQAOA and QAOA typically yield an approximation ratio of approximately 0.8 and 0.55, respectively, for both graph sizes.}
\label{fig:rqaoa}
\end{figure}

Here, we selected 10 weighted bipartite graphs with 128 and 256 vertices for each edge probability $p$ ranging from 0.1 to 1.0 increments of 0.1, where the edge weights were chosen randomly from a Gaussian distribution with a mean of 50 and a variance of 25, subsequently converted to integers. 
For the level-1 QAOA, we used the analytical form from Ref.~\cite{ODM22} to compute the expectation values for weighted bipartite graphs with a large number of vertices and optimized the parameters $(\beta,\gamma)$ using the line search strategy from Ref.~\cite{VKB+24} on $\gamma \in [0,\pi]$ with $20$ samples, followed by gradient descent, with $\beta^*=\pi/8$.
The level-1 RQAOA was optimized using a coarse grid search~\cite{VKB+24} with $20 \times 20$ points, followed by gradient descent on the best-found point.

As shown in Fig.~\ref{fig:rqaoa}, QAOA achieves an average approximation ratio $\approx 0.55$, whereas RQAOA achieves an average approximation ratio $\approx 0.8$. 
Thus, we can obtain a numerical result showing that RQAOA outperforms QAOA for all instances, but it does not appear that RQAOA can exactly solve the MAX-CUT problem. 
On the other hand, we numerically show that our modified RQAOA introduced in Algorithm~\ref{alg:modified_rqaoa} (denoted as RQAOA$^*$ in the figure) consistently achieves the optimal approximation ratio of 1 for all 200 instances. In the next subsection, we give an analytical proof showing that this is always true. 


Bipartite graphs (both unweighted and weighted) are well-structured so that the exact solution can be easily found. Surprisingly, we can see from our numerical result that both QAOA and RQAOA have limitations in solving the MAX-CUT problem even for instances like (weighted) bipartite graphs. 
 In the following subsection, we propose a modified RQAOA with a novel parameter-setting strategy to settle these issues.

\subsection{Modified RQAOA achieves the approximation ratio of 1}
\label{sec:new_rqaoa}

Normally, if we restrict the region to find the optimal parameters $(\beta^*, \gamma^*)$ to maximize the expectation value $F_q(\mathbf{\beta},\mathbf{\gamma})$ in QAOA, the maximum expectation value from QAOA is expected to decrease. We demonstrate that restricting the optimization domain in the QAOA subroutine, regardless of whether optimal parameters are found in the QAOA subroutine, surprisingly improves the performance of RQAOA in solving the MAX-CUT problem on weighted bipartite graphs.

Even when considering the performance of RQAOA in solving the MAX-CUT problem on unweighted graphs, all reduced graphs in each iteration naturally become weighted graphs. Therefore, to analyze the performance, it is crucial to thoroughly understand the performance of the QAOA subroutine on weighted graphs.

In this section, we first define a (undirected) weighted graph called a \textit{graph with parity-signed weights} 
as follows.  
\begin{Def} 
Let $G=(V,E,w)$ be a weighted graph with a vertex set $V$, an edge set $E$, and a weight function $w:E \rightarrow \mathbb{R}-\{0\}$. Suppose that there is a function $\sigma:V \rightarrow \{0,1\}$. Then $V$ can be partitioned into two subsets $V_{0}=\{v \in V : \sigma(v) = 0\}$ and $V_{1}=\{v \in V : \sigma(v) = 1\}$.
We call $G=(V,E,w,\sigma)$ as a graph with parity-signed weights~\footnote{In graph theory, a parity-signed graph, unlike our case, assigns a positive sign to edges between vertices with the same parity and a negative sign to edges between vertices with opposite parity.} if 
$$
sgn(w(e))=(-1)^{\sigma(i)+\sigma(j)+1}
$$
for all edge $e=(i,j) \in E.$

\end{Def}


\begin{Rem} 
    \begin{enumerate}
        \item The edge set $E$ of a graph with parity-signed weights $G=(V,E,w,\sigma)$ can be written as $E = E_0^- \bigcup E_1^- \bigcup E_2^+$, where $E_0^- \subseteq V_0 \times V_0$, $E_1^- \subseteq V_1 \times V_1$, and $E_2^+ \subseteq V_0 \times V_1$. By the definition of $G$, $w(e_0),w(e_1)<0$ and $w(e_2)>0$ for $e_0 \in E_0^-, e_1 \in E_1^-$ and $e_2 \in E_2^+$.    
        
        \item Every positive weighted bipartite graph can be considered as a graph with parity-signed weights because it is the case when $E_0^- = E_1^- = \emptyset$. 
             
    \end{enumerate}
\end{Rem}

We propose a modified RQAOA equipped with a novel parameter setting method and prove that it can get the exact MAX-CUT solution on a graph with parity-signed weights $G=(V,E,w,\sigma)$. 
\begin{algorithm}[Modified RQAOA$_1$]
\label{alg:modified_rqaoa}
We only modify the first step of RQAOA$_1$ and the rest part of our algorithm is the same as the original RQAOA$_1$.
\begin{enumerate}
    \item Apply the level-$1$ QAOA to find the optimal parameters $({\beta}^*,{\gamma}^*)$ to maximize $F_1(\mathbf{\beta},\mathbf{\gamma})$ in the restricted domain where $0<|\gamma| \le \frac{\pi}{2 w^*}$ and $0<\beta<\frac{\pi}{4}$ with $w^* = \max_{(i,j) \in E} |w(ij)|$.
\end{enumerate}
\end{algorithm}

\begin{Thm}
\label{thm:main}
    Our Algorithm~\ref{alg:modified_rqaoa} can exactly solve the MAX-CUT problem on a graph with parity-signed weights. Hence, it can always get the exact MAX-CUT solution on positive weighted bipartite graphs.
\end{Thm}

We provide the proof of Theorem~\ref{thm:main} in the following section.

\subsubsection{Proof of Theorem~\ref{thm:main}}

In this section, we prove that applying our parameter setting method can exactly solve the MAX-CUT problem on a graph with parity-signed weights in the following steps.
Lemma~\ref{lem:cost_preserv} states that the MAX-CUT cost function value and the weight structure of $G$ are preserved after identifying two connected vertices with the same or opposite sign when they belong to the same or different parts of the vertex set, respectively.
Lemma~\ref{lem:our1} yields that our modified RQAOA assigns edge correlation as described in Lemma~\ref{lem:cost_preserv}.
Consequently, combining Lemma~\ref{lem:cost_preserv} and Lemma~\ref{lem:our1}, we can prove that our modified RQAOA can always find the exact MAX-CUT solution for a graph with parity-signed weights.

\begin{Lem}
\label{lem:cost_preserv}
If we eliminate an edge $e$ by identifying the vertices of $e$ such that the variables of the vertices with the same or opposite sign when $e \in E_0^- \cup E_1^-$ or $e \in E_2^+$, respectively, the value of the MAX-CUT cost function on ${G=(V,E,w,\sigma)}$ and the weight structure of reduced graphs remain preserved. 
\end{Lem}

We can intuitively see that the weight structure of a graph with parity-signed weights is preserved after identifying two connected vertices for both cases as described in Fig.~\ref{fig:reduction}. Moreover, we can also figure out that the maximum value of the MAX-CUT cost function $C^w(\bold{x})$ which is defined in Eq.~\ref{eq:w_maxcut}, does not change. Note that the maximum value of $C^w(\bold{x})$ on $K_{4,3}^w$ is the sum of all positive weights, that is, the sum of all red edge weights. In Case 1, the positive weights of the removed edges are maintained by merging with the edge weights of the red bold lines. Thus, the maximum value of $C^w(\bold{x})$ remains to be unchanged. See Appendix~\ref{appendix:pf_lem} for the detailed proof.

\begin{figure}

	\centering
        \subfloat[][Case 1: Identify two vertices $i$ and $i'$ in $V_0$ with the same sign] 
       {\includegraphics[width=0.85\linewidth]{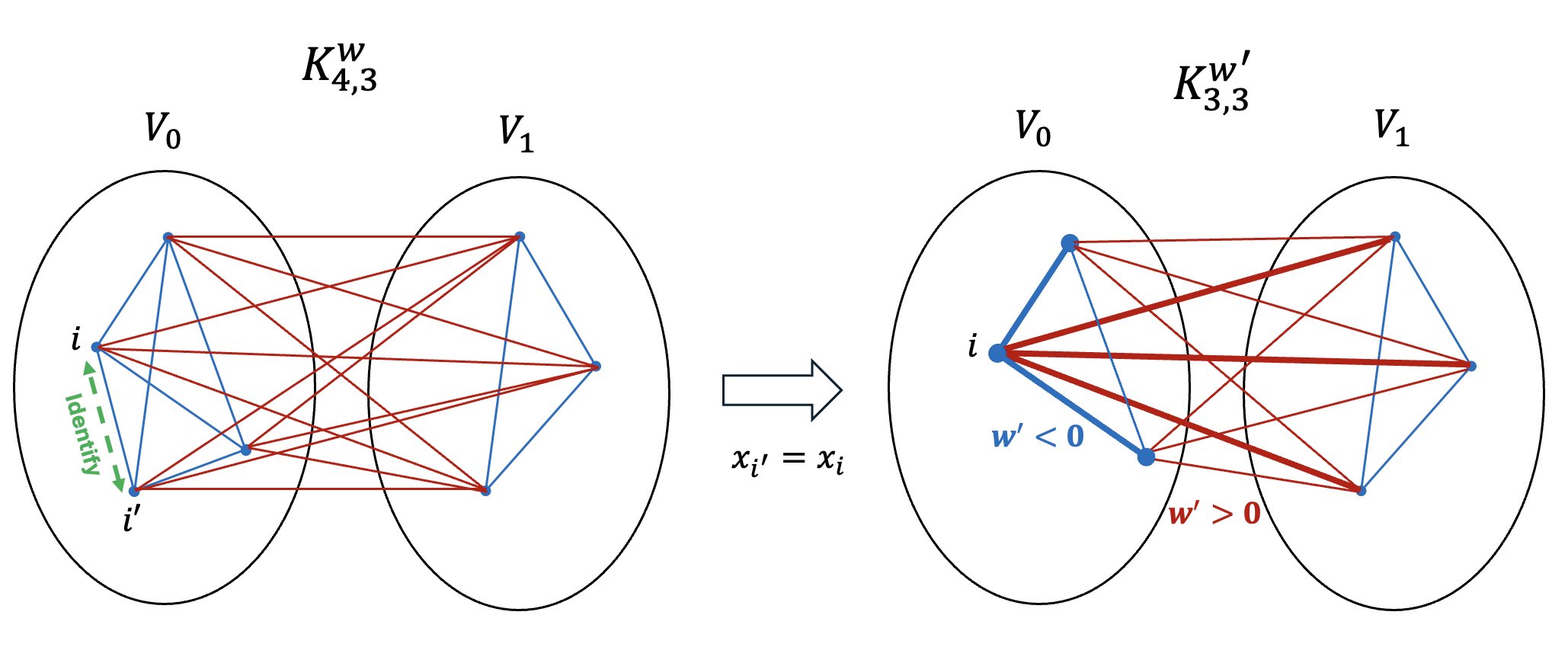}}
        \\
        \subfloat[][Case 2: Identify a vertex $i$ in $V_0$ and a vertex $j$ in $V_1$ with opposite signs] 
        {\includegraphics[width=0.85\linewidth]{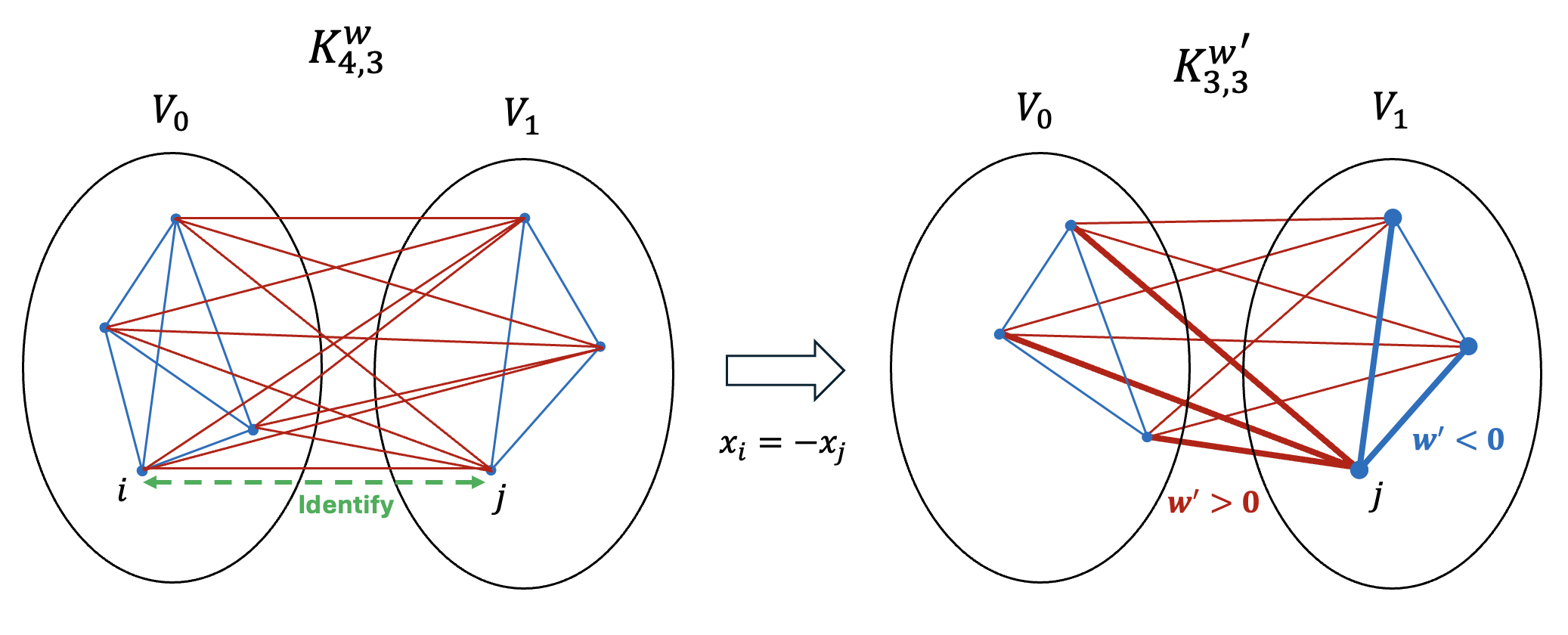}}
\caption{\justifying A schematic diagram showing how the MAX-CUT cost function on $K^w_{4,3}$ and the weight structure of the graph change after one iteration of our modified RQAOA for the case when we identify (a) two connected vertices in the same partition with the same sign and (b) in the different partitions with different signs; 
The red/blue edges indicate the edges with the positive/negative weights, respectively. The bold lines represent the edges whose weights changed after identifying two vertices. After one iteration, the reduced graph remains to have the same weight structure as shown in the right graph for both cases.}
   \label{fig:reduction}
\end{figure}

\begin{Lem}
\label{lem:our1}
     If $0 <\left|\gamma^*\right| \le \frac{\pi}{2 w^*}$ with $w^* = \max_{(i,j) \in E} |w(ij)|$ and $0<\beta^*<\frac{\pi}{4}$  in the QAOA$_1$ subroutine for solving the MAX-CUT problem on a graph with parity-signed weights $G=(V,E,w,\sigma)$, the edge expectation values are always positive or negative when the edges belong to $E_0^- \cup E_1^-$ or $E_2^+$, respectively.
\end{Lem}

We also put the detailed proof of Lemma~\ref{lem:our1} in Appendix~\ref{appendix:pf_lem:our1}. Now, we can prove Theorem~\ref{thm:main} as follows.



\begin{proof}[Proof of Theorem~\ref{thm:main}]
    The weighted MAX-CUT cost function on ${G=(V,E,w,\sigma)}$, where $V=V_0 \cup V_1$ with $|V_0|=n$ and $|V_1|=m$ and $E=E_0^- \cup E_1^- \cup E_2^+$, can be written as 
    \begin{eqnarray*}
     C_{n,m}^w(\bold{x}) &=& \frac{1}{2} \left(
     \sum_{e_2 \in E_2^+}  w(e_2)^+ 
     + \sum_{e_0 \in E_0^-}  w(e_0)^- 
     + \sum_{e_1 \in E_1^-}  w(e_1)^-
     - \sum_{e_2 \in E_2^+}  w(e_2)^+x_{e_2} 
     - \sum_{e_0 \in E_0^-}  w(e_0)^-x_{e_0} 
     - \sum_{e_1 \in E_1^-}  w(e_1)^-x_{e_1}
     \right), 
 \end{eqnarray*}
 where $w(e)^+$ and $w(e)^-$ indicate the 
 positivity and negativity of the edge weight of the edge $e$, respectively.
    It follows from Lemma~\ref{lem:cost_preserv} and Lemma~\ref{lem:our1} that we can exactly calculate the new MAX-CUT cost function after $k$ iterations. 
    Assume that there are $l$ eliminated edges with the correlated end vertices and $k-l$ eliminated edges with the anti-correlated end vertices during $k$ iterations. Let $\mathcal{X}$ be the subset of $\{-1,1\}^{n+m}$ satisfying those $k$ constraints. Then the maximum value of the new MAX-CUT cost function we can obtain by imposing the $k$ constraints is
    \begin{eqnarray*}
     \max_{x \in \mathcal{X}} C_{n,m}^w(\bold{x}) 
     &=& \sum_{e \in E^{\text{elm}}} w(e) - \sum_{e \in E^{\text{cre}}} w(e) + \max_{\bold{x}'} C_{n-k_1,m-k_2}(\bold{x}') \\
     &=& \sum_{e \in E^{\text{elm}}} w(e) - \sum_{e \in E^{\text{cre}}} w(e) + \left(\sum_{e_2 \in E_2^+} w(e_2)^+ - \sum_{e \in E^{\text{elm}}} w(e) +\sum_{e \in E^{\text{cre}}} w(e) \right) \\
     &=& \max_{\bold{x}\in\{-1,1\}^{n+m}}C_{n,m}^w(\bold{x}),
 \end{eqnarray*}
 where $E^{\text{elm}}$ and $E^{\text{cre}}$ are the set of all eliminated and created edges in the edge set $E_2^+$ during $k$ iterations, respectively, and $k=k_1+k_2$. Thus, by finding the exact MAX-CUT solution of the final reduced graph ${G'=(V',E',w',\sigma)}$, where $V'=V'_0 \cup V'_1$ with $|V'_0|=n-k_1$ and $|V'_1|=m-k_2$, the MAX-CUT cost function value of the initial graph will recover perfectly using the $k$ constraints. This completes the proof.
\end{proof}

\section{Conclusion}
\label{sec:conclusion}

 We have shown
 the upper bound of the approximation ratio of the level-1 QAOA in solving the MAX-CUT problem on bipartite graphs.
 Our findings indicate that the performance of the level-1 QAOA
 generally deteriorates as a bipartite graph becomes denser and as the size of a complete bipartite graph grows.
 
 In addition, we have numerically shown that although RQAOA outperforms QAOA at the same level in solving the MAX-CUT problem on weighted bipartite graphs, there is still room for further performance improvement. To get a better performance of RQAOA, we have proposed a modified RQAOA with a parameter setting strategy and have rigorously proved that it can exactly solve the MAX-CUT problem on a graph with parity-signed weights which is a generalization of a positive weighted bipartite graph.

 The level-1 QAOA itself has limitations in parameter optimization. Since our modified RQAOA even restricts the optimization region in the QAOA subroutine, it may not always yield true optimal parameters at each iteration. Nevertheless, our algorithm demonstrates better performance than the original RQAOA in solving the MAX-CUT problem on positive weighted bipartite graphs. What aspects of our modified RQAOA lead to its performance improvement? Unfortunately, we have not yet obtained a complete answer to this question from this work. Hence, rigorously analyzing it would be an important area for future research.
 


We have proved that our strategy restricting the optimization region yields the optimal performance of RQAOA in solving the MAX-CUT problem for graphs with parity-signed weights.
Our results support that good performance of RQAOA does not require good optimization of QAOA, although refining the grid search area more finely in the classical optimization of the QAOA subroutine to obtain a better optimal parameter is expected to improve RQAOA’s performance. 
In addition, in the case of well-structured graphs like parity-signed weighted graphs, preserving the graph structure across iterations appears to enhance the RQAOA performance.
Therefore, for future works, it would be interesting to study if our strategy could fit into more general instances such as triangle-free graphs or triangle-free $d$-regular graphs since almost all triangle-free graphs are bipartite~\cite{PSS02}.


\begin{acknowledgements}
This research was supported by  the National Research Foundation of Korea (NRF) of Korea grant funded by the Ministry of Science and ICT (MSIT) (Grant
No. NRF-2022R1C1C2006396). S.L. acknowledges support from the NRF grants funded by the MSIT (No. NRF-2022M3K2A108385913 and No. RS-2024-00432214), and Creation of the Quantum Information Science R\&D Ecosystem (No. 2022M3H3A106307411) through the NRF funded by the MSIT.
H.K. is supported by the KIAS Individual Grant No. CG085301 at Korea Institute for Advanced Study. 
E.B. and H.K. acknowledge support from the NRF of Korea (Grants No. 2023M3K5A109480511 and No. 2023M3K5A1094813). V.V. is supported by the A*STAR Start Up Fund (KIMR220701aIMRSEF) and Q.InC Strategic Research and Translational Thrust.
\end{acknowledgements}

\begin{appendices}

\section{Appendix: Proof of Lemma~\ref{lem:cost_preserv}}
\label{appendix:pf_lem}
In this section, we provide the detailed calculations in the proof of Lemma~\ref{lem:cost_preserv} which states that the reduced graph remains to preserve the weight structure of a graph with parity-signed weights. Moreover, the MAX-CUT cost function will not change after identifying two connected vertices with the same or opposite sign when they have the same or opposite parity, respectively.
The MAX-CUT cost function $C^w(x)$ on a weighted graph $G=(V,E,w)$ can be formulated as
\begin{equation}
\label{eq:w_maxcut}
    C^w(\bold{x}) = \frac{1}{2} \sum_{e \in E} w(e)(1 - x_e), 
\end{equation}
where $x_e = x_i x_j$ for $e=(i,j) \in E$ with $x_i \in \{-1,1\}$ for all $i \in V$, and $w(e)$ is the edge weight of the edge $e$. 

 Let us denote the cost function of the weighted MAX-CUT problem over a graph with parity-signed weights $G_{n,m}=(V_0 \cup V_1,E_0^- \cup E_1^- \cup E_2^+, w, \sigma)$ with $|V_0|=n$ and $|V_1|=m$ by
 \begin{eqnarray}
 \label{eq:w_maxcut_n,m}
     C_{n,m}^w(\bold{x}) &=&
     \frac{1}{2} \left(
     \sum_{k \in V_0} \sum_{l \in V_1} w(kl)
     + \sum_{k \neq k'\in V_0}  w(kk')
     + \sum_{l \neq l' \in V_1}  w(ll') \nonumber \right. \\ 
     &&
     \left.-\sum_{k \in V_0} \sum_{l \in V_1} w(kl)x_k x_l 
     - \sum_{k \neq k'\in V_0}  w(kk') x_{k}x_{k'} 
     - \sum_{l \neq l' \in V_1}  w(ll')  x_{l}x_{l'} \right),
 \end{eqnarray}
 where $w(kl)$ denotes the edge weight of the edge $(k,l)$.

    Case~1) Assume that we pick $e^* \in E_0^- \cup E_1^-$. Without loss of generality, let $e^* = (i,i') \in E_0^-$. 
 We can rewrite the cost function as
 \begin{eqnarray}
 \label{eq:w_maxcut_n,m_exp}
     C_{n,m}^w(\bold{x}) &=& 
     \frac{1}{2} \left(
     \sum_{k \in V_0} \sum_{l \in V_1} w(kl) 
     + \sum_{k \neq k'\in V_0}  w(kk') 
     + \sum_{l \neq l' \in V_1}  w(ll') 
     - \sum_{k \in V_0-\{i'\}} \sum_{l \in V_1} w(kl)x_k x_l \nonumber \right.\\ 
     && \left.
     - \sum_{l \in V_1} w(i'l)x_{i'} x_l 
     - \sum_{k \neq k' \in V_0-\{i'\}}  w(kk') x_{k}x_{k'} 
     - \sum_{k \in V_0-\{i'\}}  w(ki') x_{k}x_{i'} 
     - \sum_{l \neq l' \in V_1}  w(ll')  x_{l}x_{l'} \right).
 \end{eqnarray}
    We impose the constraint $x_{i'}=x_{i}$ on the cost function $C_{n,m}^w(\bold{x})$ to get the new cost function
    
 \begin{eqnarray}
 \label{eq:w_maxcut_n,m_expension}
     C_{n,m}^w(\bold{x}) 
     &=& 
     \frac{1}{2} \left(
     \sum_{k \in V_0} \sum_{l \in V_1} w(kl) 
     + \sum_{k \neq k'\in V_0}  w(kk') 
     + \sum_{l \neq l' \in V_1}  w(ll') 
     - \sum_{k \in V_0-\{i'\}} \sum_{l \in V_1} w(kl)x_k x_l \nonumber \right.\\ 
     && \left.
     - \sum_{l \in V_1} w(i'l)x_{i} x_l 
     - \sum_{k \neq k'\in V_0-\{i'\}}  w(kk') x_{k}x_{k'} 
     - \sum_{k \in V_0-\{i'\}}  w(ki') x_{k}x_{i} 
     - \sum_{l \neq l' \in V_1}  w(ll')  x_{l}x_{l'}  \right) \\
    &=& 
     \frac{1}{2} \left(
     \sum_{k \in V_0} \sum_{l \in V_1} w(kl) 
     + \sum_{k \neq k'\in V_0}  w(kk') 
     + \sum_{l \neq l' \in V_1}  w(ll') 
     - \sum_{k \in V_0-\{i'\}} \sum_{l \in V_1} w(kl)x_k x_l \nonumber \right.\\ 
     && \left.
     - \sum_{l \in V_1} w(i'l)x_{i} x_l 
     - \sum_{k \neq k' \in V_0-\{i'\}}  w(kk') x_{k}x_{k'} 
     - \sum_{k \in V_0-\{i,i'\}}  w(ki') x_{k}x_{i} 
     - w(ii')
     - \sum_{l \neq l' \in V_1}  w(ll')  x_{l}x_{l'}  \right) \\
     &=& 
     \frac{1}{2} \left( 
     \sum_{k \in V_0} \sum_{l \in V_1} w(kl) 
     + \sum_{k \neq k'\in V_0}  w(kk') 
     + \sum_{l \neq l' \in V_1}  w(ll') 
     \underbrace{- \sum_{k \in V_0-\{i'\}} \sum_{l \in V_1} w(kl)
     - \sum_{l \in V_1} w(i'l)}
     _{= -\sum_{k \in V_0} \sum_{l \in V_1} w(kl)}
     \nonumber \right.\\ 
     && \left.
     \underbrace{- \sum_{k \neq k' \in V_0-\{i'\}}  w(kk') 
     - \sum_{k \in V_0-\{i,i'\}}  w(ki')
     - w(ii')}
     _{= -\sum_{k \neq k' \in V_0}  w(kk')}
     - \sum_{l \neq l' \in V_1}  w(ll')   \right) + C_{n-1,m}^w(\bold{x}') \\
     &=& C_{n-1,m}^w(\bold{x}').
 \end{eqnarray}
 Here, $C_{n-1,m}^w(\bold{x}')$ is the MAX-CUT cost function of the reduced graph ${G_{n-1,m}=(V_0 -\{i'\} \cup V_1, E',w',\sigma)}$ which can be formulated as
 $$
  C_{n-1,m}^w(\bold{x}')= \frac{1}{2} \sum_{e \in E'} w'(e)(1 - x'_e),
 $$
 where $E'=E-\{e^*\}$, $\bold{x}'=(x_1,\dots, \hat{x_{i'}}, \dots, x_{n+m}) \in \{-1,1\}^{n+m-1}$ is the vector $\bold{x} \in \{-1,1\}^{n+m}$ with the component $x_{i'}$ excluded, and 
 $$
 w'(e) = 
 \begin{cases}
    w(e)
     & \mathrm{if}~ i \notin e \in E, \\
    w(ki) + w(ki') & \mathrm{if}~e=(k,i)\in E_0^-~ \mathrm{for}~ k \in V_0 - \{i,i'\}, \\
    w(il) + w(i'l) & \mathrm{if}~e=(i,l) \in E_2^+~ \mathrm{for}~ l \in V_1.
    \end{cases}
 $$
Thus, the value of the MAX-CUT cost function does not change after identifying two correlated vertices with the same sign. Since $w(ki)^- + w(ki')^- <0$ and $ w(il)^+ + w(i'l)^+>0$, the reduced graph $G_{n-1,m}$ is also a graph with parity-signed weights as shown in Fig.~\ref{fig:reduction}.

  Case~2) Assume that we pick $e^* \in E_2^+$. Without loss of generality, let $e^* = (i,j) \in E_2^+$. 
 Similarly, we can rewrite the MAX-CUT cost function as
 \begin{eqnarray}
 \label{eq:w_maxcut_n,m_exp2}
     C_{n,m}^w(\bold{x}) &=& 
     \frac{1}{2} \left(
     \sum_{k \in V_0} \sum_{l \in V_1} w(kl) 
     + \sum_{k \neq k'\in V_0}  w(kk') 
     + \sum_{l \neq l' \in V_1}  w(ll') 
     - \sum_{k \in V_0-\{i\}} \sum_{l \in V_1} w(kl)x_k x_l \nonumber \right.\\ 
     && \left.
     - \sum_{l \in V_1} w(il)x_{i} x_l 
     - \sum_{k\neq k' \in V_0-\{i\}}  w(kk') x_{k}x_{k'} 
     - \sum_{k \in V_0-\{i\}}  w(ki) x_{k}x_{i} 
     - \sum_{l \neq l' \in V_1}  w(ll')  x_{l}x_{l'} \right).
 \end{eqnarray}
    We impose the constraint $x_{i}= -x_{j}$ on the cost function $C_{n,m}^w(\bold{x})$ to get the new cost function

 \begin{eqnarray}
 \label{eq:w_maxcut_n,m_expension2}
     C_{n,m}^w(\bold{x}) 
     &=& 
     \frac{1}{2} \left(
     \sum_{k \in V_0} \sum_{l \in V_1} w(kl) 
     + \sum_{k \neq k'\in V_0}  w(kk') 
     + \sum_{l \neq l' \in V_1}  w(ll') 
     - \sum_{k \in V_0-\{i\}} \sum_{l \in V_1} w(kl)x_k x_l \nonumber \right.\\ 
     && \left.
     - \sum_{l \in V_1} w(il)(-x_{j}) x_l 
     - \sum_{k \neq k'\in V_0-\{i\}}  w(kk') x_{k}x_{k'} 
     - \sum_{k \in V_0-\{i\}}  w(ki) x_{k}(-x_{j}) 
     - \sum_{l \neq l' \in V_1}  w(ll')  x_{l}x_{l'}  \right) \\
     &=& 
     \frac{1}{2} \left(
     \sum_{k \in V_0} \sum_{l \in V_1} w(kl) 
     + \sum_{k \neq k'\in V_0}  w(kk') 
     + \sum_{l \neq l' \in V_1}  w(ll') 
     - \sum_{k \in V_0-\{i\}} \sum_{l \in V_1} w(kl)x_k x_l
     - \sum_{l \in V_1-\{j\}} (-w(il))x_{j} x_l 
     \nonumber \right.\\ 
     && \left.
     - (-w(ij))
     - \sum_{k \neq k'\in V_0-\{i\}}  w(kk') x_{k}x_{k'} 
     - \sum_{k \in V_0-\{i\}}  (-w(ki)) x_{k}x_{j} 
     - \sum_{l \neq l' \in V_1}  w(ll')  x_{l}x_{l'}  \right) \\
     &=& 
     \frac{1}{2} \left( 
     \sum_{k \in V_0} \sum_{l \in V_1} w(kl) 
     + \sum_{k \neq k'\in V_0}  w(kk') 
     + \sum_{l \neq l' \in V_1}  w(ll') 
     \underbrace{- \sum_{k \in V_0-\{i\}} \sum_{l \in V_1} w(kl)
     - \sum_{l \in V_1-\{j\}} (-w(il))- (-w(ij)) }
     _{= -\sum_{k \in V_0} \sum_{l \in V_1} w(kl) + 2\sum_{l \in V_1} w(il)}
     \nonumber \right.\\ 
     && \left.
     \underbrace{- \sum_{k \neq k' \in V_0-\{i\}}  w(kk') 
     - \sum_{k \in V_0-\{i\}} (-w(ki))}
     _{= -\sum_{k \neq k' \in V_0}  w(kk') + 2\sum_{k \in V_0-\{i\}} w(ki)}
     - \sum_{l \neq l' \in V_1}  w(ll')   \right) + C_{n-1,m}^w(\bold{x}') \\
     &=& \sum_{l \in V_1} w(il) + \sum_{k \in V_0-\{i\}} w(ki) + C_{n-1,m}^w(\bold{x}').
 \end{eqnarray} 
Here,  $C_{n-1,m}^w(\bold{x}')$ is the MAX-CUT cost function of the reduced graph ${G_{n-1,m}=(V_0 -\{i\} \cup V_1, E',w',\sigma)}$  which can be formulated as
 $$
  C_{n-1,m}^w(\bold{x}')= \frac{1}{2} \sum_{e \in E'} w'(e)(1 - x'_e),
 $$
 where $E'=E-\{e^*\}$, $\bold{x}'=(x_1,\dots, \hat{x_{i}}, \dots, x_{n+m}) \in \{-1,1\}^{n+m-1}$ is the vector $\bold{x} \in \{-1,1\}^{n+m}$ with the component $x_{i}$ excluded, and 
 $$
 w'(e) = 
 \begin{cases}
    w(e)
     & \mathrm{if}~ i \notin e \in E, \\
    w(kj)-w(ki) & \mathrm{if}~e=(k,j) \in E_2^+~ \mathrm{for}~  k \in V_0 - \{i\}, \\
     -w(il) + w(jl) & \mathrm{if}~e=(j,l) \in E_1^-~ \mathrm{for}~ l \in V_1 - \{j\}.
    \end{cases}
 $$

 Note that the sum of edge weights $\sum_{l \in V_1} w(il) + \sum_{k \in V_0-\{i\}} w(ki)$ in the new cost function ${C_{n,m}^w}'(\bold{x})$ exactly corresponds to the total sum of edge weights gained or lost in the edge set $E_2^+$ when the edge $(i,j)$ is removed.
 Therefore, we have also proved that the maximum value of the cost function will not change after identifying two anti-correlated vertices with opposite sign. Moreoever, since $ w(kj)^+-w(ki)^->0$ and $-w(il)^+ + w(jl)^- <0$, the reduced graph $G_{n-1,m}$ is also a graph with parity-signed weights as shown in Fig.~\ref{fig:reduction}.

\section{Appendix: Proof of Lemma~\ref{lem:our1}}
\label{appendix:pf_lem:our1}
In this section, we provide the detailed proof of Lemma~\ref{lem:our1}. For simplicity, we will denote the edge weight $w(ij)$ as $w_{ij}$ for an edge $e=(i,j)$ in this section.
\begin{proof}
    We can use the analytic form of the edge expectation values in the level-$1$ QAOA subroutine for weighted graphs~\cite{VDKAL24}. Let us denote $\mathcal{N}_i$ and $\mathcal{T}_{ij}$ as the set of vertices that are adjacent to a vertex $i$ and form a triangle with the edge $(i,j)$, respectively. Then the edge expectation value of the edge $(i,j)$ is 
    \begin{eqnarray*}
        \left< Z_iZ_j \right> &=& -\frac{1}{2}\sin(4\beta)\sin(w_{ij}\gamma)\left( \prod_{k\neq j \in \mathcal{N}_i}\cos(w_{ki}\gamma) + \prod_{l\neq i \in \mathcal{N}_j} \cos(w_{lj}\gamma) \right) \\
        && - \frac{1}{2}\sin^2(2\beta) \prod_{\substack{k\neq j \in \mathcal{N}_i \\ k \notin \mathcal{T}_{ij}}}\cos(w_{ki}\gamma)  \prod_{\substack{l\neq j \in \mathcal{N}_j \\ l \notin \mathcal{T}_{ij}}}\cos(w_{lj}\gamma) 
        \left[ \prod_{t\in \mathcal{T}_{ij}}\cos\left((w_{it} + w_{tj})\gamma\right) 
        - \prod_{t \in \mathcal{T}_{ij}}\cos\left((w_{it} - w_{tj})\gamma\right)
        \right].
    \end{eqnarray*}
    Case~1) Assume that we pick $e^*=(i^*,j^*) \in E_0^- \cup E_1^-$ to be eliminated. Without loss of generality, let $e^* \in E_0^-$. Let $E_{i^*}=\{(i^*,k)\in E : k \neq j^*\}$ be the set of edges connected to the vertex $i^*$ excluding the edge $e^*$, and $F_{i^*}=\{(i^*,k)\in E : k \in \mathcal{T}_{i^*j^*}\}$ be the set of edges connected to the vertex $i^*$, together with the edge $e^*$, form triangles.  
    Now, we can rewrite the edge expectation value in terms of the edges with these notations as 
    \begin{eqnarray*}
        \left< Z_{i^*}Z_{j^*} \right> &=& -\frac{1}{2}\sin(4\beta)\sin(w_{e^*}^-\gamma)
        \underbrace{\left( \prod_{e_0 \in E_{i^*}\cap E_0^-}\cos(w_{e_0}^-\gamma) \prod_{e_2 \in E_{i^*}\cap E_2^+}\cos(w_{e_2}^+\gamma)
        + \prod_{e_0' \in E_{j^*}\cap E_0^-}\cos(w_{e_0'}^-\gamma) \prod_{e_2' \in E_{j^*}\cap E_2^+}\cos(w_{e_2'}^+\gamma)
        \right)}_{(*)} \\
        && - \frac{1}{2}\sin^2(2\beta) 
        \underbrace{\prod_{\substack{e_0 \in E_{i^*}\cap E_0^- \\ e_0 \notin F_{i^*}}}\cos(w_{e_0}^-\gamma)  
        \prod_{\substack{e_2 \in E_{i^*}\cap E_2^+ \\ e_2 \notin F_{i^*}}}\cos(w_{e_2}^+\gamma)
        \prod_{\substack{e_0' \in E_{j^*}\cap E_0^- \\ e_0' \notin F_{j^*}}}\cos(w_{e_0'}^-\gamma)  
        \prod_{\substack{e_2' \in E_{j^*}\cap E_2^+ \\ e_2' \notin F_{j^*}}}\cos(w_{e_2'}^+\gamma)}_{(**)} \\
        && \times \underbrace{\left[ \
        \prod_{\substack{f_0 \in F_{i^*}\cap E_0^- \\ f_2 \in F_{i^*}\cap E_2^+}}
        \cos((w_{f_0}^- + w_{\tilde{f_0}}^-)\gamma) \cos((w_{f_2}^+ + w_{\tilde{f_2}}^+)\gamma)
        - \prod_{\substack{f_0 \in F_{i^*}\cap E_0^- \\ f_2 \in F_{i^*}\cap E_2^+}}
        \cos((w_{f_0}^- - w_{\tilde{f_0}}^-)\gamma) \cos((w_{f_2}^+ - w_{\tilde{f_2}}^+)\gamma)
        \right]}_{(***)},
    \end{eqnarray*}
    where $\tilde{f}$ represents the third edge of the triangle formed by the edges $e^*$ and $f$. We denote the negative and positive edge weight of the edge $e$ by $w_e^-$ and $w_e^+$, respectively. We can easily see that if $0 < \left|\gamma \right| \le \frac{\pi}{2 w_{ij}^*}$ with $w_{ij}^* = \max_{(i,j) \in E} |w_{ij}|$, the terms $(*)$ and $(**)$ are always positive while the term $(***)$ is negative because the cosine function is decreasing on that region. Thus, the edge expectation values of the edges belonging to $E_0^-$ (and $E_1^-$) are always positive.

    Case~2) Assume that we pick $e^*=(i^*,j^*) \in E_2^+$ to be eliminated. We can apply a similar argument to the above case. The main differences will be the sign of $\sin(w_{e^*}^+\gamma)$ and the term $(***)$. Let us rewrite the edge expectation value for this case as
    \begin{eqnarray*}
        \left< Z_{i^*}Z_{j^*} \right> &=& -\frac{1}{2}\sin(4\beta)\sin(w_{e^*}^+\gamma)
        \underbrace{\left( \prod_{e_0 \in E_{i^*}\cap E_0^-}\cos(w_{e_0}^-\gamma) \prod_{e_2 \in E_{i^*}\cap E_2^+}\cos(w_{e_2}^+\gamma)
        + \prod_{e_1 \in E_{j^*}\cap E_1^-}\cos(w_{e_1}^-\gamma) \prod_{e_2' \in E_{j^*}\cap E_2^+}\cos(w_{e_2'}^+\gamma)
        \right)}_{(*)} \\
        && - \frac{1}{2}\sin^2(2\beta) 
       \underbrace{\prod_{\substack{e_0 \in E_{i^*}\cap E_0^- \\ e_0 \notin F_{i^*}}}\cos(w_{e_0}^-\gamma)  
        \prod_{\substack{e_2 \in E_{i^*}\cap E_2^+ \\ e_2 \notin F_{i^*}}}\cos(w_{e_2}^+\gamma)
        \prod_{\substack{e_1 \in E_{j^*}\cap E_1^- \\ e_1 \notin F_{j^*}}}\cos(w_{e_2}^-\gamma)  
        \prod_{\substack{e_2' \in E_{j^*}\cap E_2^+ \\ e_2' \notin F_{j^*}}}\cos(w_{e_2'}^+\gamma)}_{(**)} \\
        && \times \underbrace{\left[ \
        \prod_{\substack{f_0 \in F_{i^*}\cap E_0^- \\ f_2 \in F_{i^*}\cap E_2^+}}
        \cos((w_{f_0}^- + w_{\tilde{f_0}}^+)\gamma) \cos((w_{f_2}^+ + w_{\tilde{f_2}}^-)\gamma)
        - \prod_{\substack{f_1 \in F_{i^*}\cap E_1^- \\ f_2 \in F_{i^*}\cap E_2^+}}
        \cos((w_{f_1}^- - w_{\tilde{f_1}}^+)\gamma) \cos((w_{f_2}^+ - w_{\tilde{f_2}}^-)\gamma)
        \right]}_{(***)}.
    \end{eqnarray*}
Suppose that $0 < \left|\gamma \right| \le \frac{\pi}{2 w^*}$ with $w^* = \max_{(i,j) \in E} w_{ij}$. Similar to Case 1), the terms $(*)$ and $(**)$ are always positive. However, in contrast to Case 1), $w_{e^*} > 0$ so that $\sin(w_{e^*}\gamma)$ has the opposite sign, and the term $(***)$ is positive in this case because the cosine function is decreasing on that region together with the fact that
$$|(w_{f_0}^- + w_{\tilde{f_0}}^+)\gamma| \le |(w_{f_1}^- - w_{\tilde{f_1}}^+)\gamma|
~~\text{and}~~
|(w_{f_2}^+ + w_{\tilde{f_2}}^-)\gamma| \le |(w_{f_2}^+ - w_{\tilde{f_2}}^-)\gamma|.
$$
Therefore, the edge expectation values of the edges belonging to $E_2^+$ are always negative.
\end{proof}
 
\end{appendices}

\bibliography{RQAOA}

\end{document}